\documentclass{article}
\usepackage{amsmath}
\usepackage{amsthm}
\usepackage{indentfirst}
\usepackage{graphicx}
\usepackage{epsfig}
\usepackage[body={16cm,24cm}, top=3cm]{geometry}
\newtheorem{definition}{Definition}
\newtheorem{theorem}{Theorem}

\newtheorem{corollary}{Corollary}
\newtheorem{example}{Example}
\newtheorem{lemma}{Lemma}
\usepackage{amssymb}
\usepackage{indentfirst}
\author{Wang Changlong, Shigang Yue, Jigen Peng}
\title{When is P such that $l_0$-minimization Equals to $l_p$-minimization}
\begin{document}\large
\title{When is P such that $l_0$-minimization Equals to $l_p$-minimization}
\date{}
\maketitle
\begin{abstract}
In this paper, we present an analysis expression of $p^{\ast}(A,b)$ such that the unique solution to $l_0$-minimization also can be the unique solution to $l_p$-minimization for any $0<p<p^{\ast}(A,b)$. Furthermore, the main contribution of this paper isn't only the analysis expressed of such $p^{\ast}(A,b)$ but also its proof. Finally, we display the results of two examples to confirm the validity of our conclusions
\end{abstract}
\textbf{keywords:} sparse recovery, null space constant, null space property, $l_0$-minimization, $l_p$-minimization,
\section{INTRODUCTION}
One of the core problems in Compressed Sensing is to find the sparest solution to the underdetermined system $Ax=b$, where $A \in R^{n \times m}$ is an underdetermined matrix (i.e. $n<m$), and $b \in R^m$ a vector representing some measured or represented signals. The problem is popularly modeled into the following $l_0$-minimization:
 \begin{eqnarray}
 \mathop{\min}\limits_{x \in R^m} \|x\|_0\ s.t. \ Ax=b
\end{eqnarray}
where $\|x\|_0$ indicates the number of nonzero elements of $x$, which is commonly called $l_0$-norm although it is not a real vector norm. Since $A$ has more columns than rows, the underdetermined linear system $Ax=b$ admits an infinite number of solutions. In order to find the sparest one, much excellent theoretical work has been devoted themselves to the $l_0$-minimization (1).

However, in paper [1] the author proves that $l_0$-minimization (1) is NP-hard and is combinational and computationally intractable because of the discrete and discontinuous nature. Therefore, alternative strategies to find sparest solution have been put forward (see, for example [3][5][11][12][17][22][23]), among which is the model $l_1$-minimization:
\begin{eqnarray}
\mathop{\min}\limits_{x \in R^m} \|x\|_1\ s.t. \ Ax=b
\end{eqnarray}
where $\|x\|_1=\sum_{i=1}^m |x_i|$, the $l_1$-norm of $x$. The model (2) is a convex problem and hence can be recast as linear problem solvable efficiently. However, in order to get the sparest solution to $Ax=b$ by $l_1$-minimization (1) we need certain conditions on $A$ and/or $b$, for example, the novel Restricted Isometry Property (RIP) of $A$. A matrix $A$ is said to have restricted isometry property of order s if there exists a constant $0<\delta_s<1$ such that
\begin{eqnarray}
1-\delta_s \leq \displaystyle\frac{\|Ax\|_2}{\|x\|_2} \leq 1+\delta_s
\end{eqnarray}
for any s-sparse vector $x$. A vector $x$ is said s-sparse if $\|x\|_0\leq s $.

In papers [3] and [5] Cand\.{e}s and Tao show that any s-sparse vector can be recovered via $l_1$-minimization (2) as long as $\delta_{3s}+\delta_{4s}<2$ or $\delta_{2s}<{\sqrt{2}-1}$. In paper [11], the author improves the latter inequality and establishes exact recovery of s-sparse vector via $l_1$-minimization under the condition $\delta_{2s}<2(3-\sqrt{2})/7$.

However, it should be pointed out that the problem of calculating $\delta_{2s}$ of a given matrix $A$ is still NP-hard, Work done by Gribuval and Nielsen [10] adopts a new strategy that lies between $l_0$-minimization (1) and $l_1$-minimization (2). Instead of $l_1$-minimization (3), they consider the $l_p$-minimization with $0<p<1$
\begin{eqnarray}
\mathop{\min}\limits_{x \in R^m} \|x\|_p^p \ s.t. \ Ax=b
\end{eqnarray}
where $\|x\|_p^p=\sum_{i=1}^m |x_i|^p$. In the literature, $\|x\|_p$ is still called $l_p$-norm of $x$ though it is only a quass-norm when $0<p<1$ (because in this case it violates the triangular inequality).
From the definition of $l_p$-norm, it seems to be more natural to consider $l_p$-minimization (4) instead of $l_0$-minimization (1) than others, due to the fact that $\|x\|_0=\mathop{lim}\limits_{p \to 0} \|x\|_p^p$. In paper [4], Chartrand claims that an s-sparse vector can be recovered by $l_p$-minimization (4) for some $p > 0$ small enough provided $\delta_{2s+1} < 1$.

Recently, Peng, Yue and Li [7] proves that there exists a constant $p(A,b)>0$, such that every a solution to $l_p$-minimization is also the solution to $l_0$-minimization whenever $p<p(A,b)$. This result builds a bridge between $l_p$-minimization (4) and $l_0$-minimization (1), and what is important is that this conclusion is not limited by the structure of matrix. However, the authors do not give an analysis expression of $p(A,b)$ so that the model of choice of $l_p$-minimization (4) is still difficult. That is, when does $l_0$-minimization (1) equal to $l_p$-minimization (4) is still open. In this paper we devote ourselves to giving a complete answer to this problem.

Our paper is organized as follows. In Section 2, we will present some preliminaries of the null space property, which plays a core role in the proof of our main theorem. In Section 3 we focus on proving the main results of this paper. There we will present an analysis expression of $p(A,b)$ such that the unique solution of $l_0$-minimization (1) is also the unique solution of $l_p$-minimization (4). Finally, we summarize our finding in last section.

For convenience, for $x \in R^m$, we define its support by $support\ (x)=\{i:x_i \neq 0\}$ and the cardinality of set S by $|S|$.
Let $Ker(A)=\{x \in R^m:Ax=0\}$ be the null space of matrix A, denote by $\lambda_{min^{+}}(A)$ the minimum nonzero absolute-value eigenvalue of A and by $\lambda_{max}(A)$ the maximum one. We also use the subscript notation $x_S$ to denote such a vector that is equal to $x$ on the index set S and zero everywhere else. and use the subscript notation $A_S$ to denote a submatrix whose columns are those of the columns of $A$ that are in the set index S.
\section{PRELIMINARISE}
In order to investigate conditions under which both $l_0$-minimization (1) and $l_p$-minimization (4) have the same unique solution, it is convenient for us to work with a sufficient and necessary condition of the solution to $l_0$-minimization (1) and $l_p$-minimization (4). Therefore, in this preliminary section, we focus on introducing such an condition, namely the null space property (NSP).
\begin{theorem}
(NSP)[10].Given a matrix $A \in R^{n \times m}$ with $n \leq m$, $x^{\ast}$ is the unique solution to $l_p$-minimization $(0\leq p \leq 1)$  if and only if:
\begin{eqnarray}
\|x_S\|_p < \|x_{S^C}\|_p
\end{eqnarray}
 for any $x \in Ker(A)$, and set $S$ with $|S| \leq |T^{\ast}|$, where $T^{\ast}=support(x^{\ast})$
\end{theorem}

NSP provides a sufficient and necessary condition to judge a vector whether can be recovered by $l_0$-minimization (1) or $l_p$-minimization (4), which is the most important advantage of NSP. However, NSP is difficult to be checked for a given matrix. In order to reach our goal, we recall the concept null space constant (NSC), which is closely related to NSP and will offer tremendous help in illustrating the performance of $l_0$-minimization (1) and $l_p$-minimization (4).
\begin{definition}
(NSC)[9].Let $0 \leq p \leq 1$ and $t>0$, the null space constant $h(p,A,t)$ is the smallest number such that:
$$\sum_{i \in S}|x_i|^p \leq h(p,A,t) \sum_{i \notin S}|x|^p$$
for any index set $S \subset \{1,2,\ldots,m\}$ with $|S| \leq t$ and any $x \in Ker(A)$
\end{definition}

Combining NSC, NSP and papers [2] [21], we can derive the following corollaries:
\begin{corollary}
For any $p \in [0,1]$, $h(p,A,t)<1$ is a sufficient and necessary condition to recovery any t-sparse vector by $l_p$-minimization
\end{corollary}
\begin{corollary}
If the condition (5) is satisfied for some $0<p^{\ast}\leq 1$, then it is also satisfied for all the $0\leq p \leq p^{\ast}$ [2],[21].
\end{corollary}
\begin{corollary}
In the case when $p=0$, if $x^{\ast}$ is the unique solution to the $l_0$-minimization (1), and $\|x^{\ast}\|_0=t$ , then we have the following results:

(a) $\|x\|_0 \geq 2t+1$, for any $x \in Ker(A)$.

(b) $\displaystyle t \leq \lceil \frac{m-2.5}{2} \rceil+1$, where $\lceil a \rceil$ represents the integer part of $a$
\end{corollary}

In paper [9], the author proves that $h(p,A,t)$ is a continuous function in $p\in [0,1]$ when $t\leq spark(A)-1$, where $spark(A)$ is the smallest number of columns from $A$ that are linearly dependent. Therefore, if $h(0,A,t)<1$ for some fixed $A$ and $t$, then there exists a constant $p^{\ast}$ such that $h(p,A,t)<1$ for any $p \in[0,p^{\ast}]$, i.e. both $l_0$-minimization (1) and $l_p$-minimization (4) have the same unique solution. This a corollary of the main theorem in [7].

\begin{theorem}
$[7]$ There exists a constant $p(A,b)>0$ such that, when $0<p<p(A,b)$, every solution to $l_p$-minimization (4) also solves $l_0$-minimization (1).
\end{theorem}
Obviously, this theorem qualitatively proves the effectiveness of solving the original $l_0$-minimization (1) problem through $l_p$-minimization (4). Moreover, the theorem will become more practical, if $p^{\ast}$ is computable. Since the main aim of this paper is to give a computable method for the estimation of $p(A,b)$, we can conclude the following lemma.
\begin{lemma}
Let $A \in R^{n \times m}$ be an underdetermined matrix, if $l_0$-minimization (1) has the unique solution $x^{\ast}$ with $\|x^{\ast}\|_0=t$, then there exists a constant $u>0$ with $u^2 \geq \lambda_{min^{+}}(A^TA)$ such that $\|Az\|_2 \geq u\|z\|_2$, for any $z \in R^m$ with $\|z\|_0 \leq 2t$
\end{lemma}
\begin{proof}
The proof is divided into two steps.

Step 1: To prove the existence of $u$

In order to prove this result we just need to prove that the set $S=\{u:  \|Az\|_2/\|z\|_2 \geq u, for\ any\ z \neq 0\ with\ \|z\|_0 \leq 2t\}$ has a nonzero infimum.

If we assume that inf S=0, i.e. for any $n \in N^{+}$, there exists a vector $\|z_n\|_0 \leq 2t$ such that $\|Az_n\|_2/\|z_n\|_2 \leq n^{-1}$. Without of generality, we can assume $\|z_n\|_2=1$, furthermore, the bounded sequence $\{z_n\}$ has a subsequence $\{z_{n_i}\}$ which is convergent, i.e. $z_{n_i} \to z_0$ and it is obvious that $Az_0=0$ because that the function $y(x)=Ax$ is a continuous one.

Let $J(z_0)=\{i:(z_0)_i \ne 0\}$, since $z_{n_i} \to z_0$, it is easy to get that: for any $i \in J(z_0)$, there exists $N_i$ such that $(z_{n_k})_i \ne 0$ when $k \geq N_i$.

Let $ N=\max \limits_{i \in J(z_0)}N_i$, when $k \geq N$, for any $i \in J(z_0)$, such that $(z_{n_k})_i \ne 0$. Therefore, we can get that $\|z_{n_k}\|_0 \geq \|z_0\|_0$ when $k \geq N$, such that $\|z_0\|_0 \leq 2t$.

However, according to Corollary 3, it is easy to get that $\|x\|_0 \geq 2t+1$ for any $x\in Ker(A)$. and we notice that $z_0 \in Ker(A)$, so the result $\|z_0\|_0 \leq 2t$ contradicts Corollary 3.

Therefore, there exists a constant $u>0$ such that $\|Az\|_2 \geq u\|z\|_2$, for any $z \in R^m$ with $\|z\|_0 \leq 2t$.

Step 2: To prove $u^2 \geq \lambda_{min^{+}}(A^TA)$

According to the proof above, there exists a vector $\|\widetilde x\|_0 \leq 2t$ such that $\|A\widetilde x\|_2=u\| \widetilde x\|_2$.

Let $V=support (\widetilde x)$, it is easy to get that $u^2x^Tx \leq x^TA_V^TA_Vx$, for any $x \in R^{|V|}$. It is obvious that $A_V^TA_V \in R^{|V| \times |V|}$ is a symmetric matrix, according to the Rayleigh-Ritz theorem, the smallest eigenvalue of $A_V^TA_V$  is $u^2$ and we can choose an eigenvector $z\in R^{|V|}$ of eigenvalue $u^2$.

If $u^2<\lambda_{min^{+}}(A^TA)$, then we can consider such a vector $x^{'}\in R^m$ with $x_i^{'}=z_i$ when $i \in V$ and zero everywhere else. Therefore, it is easy to get that $A^TAx^{'}=u^2x^{'}$ which contradicts the definition of $\lambda_{min^{+}}(A^TA)$

Therefore, the proof is completed.
\end{proof}
\begin{corollary}
Let $A \in R^{n \times m}$ be an underdetermined matrix, if $l_0$-minimization (1) has the unique solution $x^{\ast}$ with $\|x^{\ast}\|_0=t$, there exist constants $0<u\leq w \leq \lambda_{max}(A^TA)$ such that $u\|x\|_2^2 \leq \|Ax\|_2^2 \leq w\|x\|_2^2 $ for any $\|x\|_0 \leq 2t$.
\end{corollary}

\section{MAIN CONTRIBUTION}
In this section we focus on the proposed problem in the last section. By introducing a new technique and utilizing preparations provided in Section 2, we will present an analysis expression of $p^{\ast}(A,b)$ such that both $l_0$-minimization(1) and $l_p$-minimization (4) have the same unique solution for $0<p<p^{\ast}(A,b)$. To this end, we first begin with two lemmas.

\begin{lemma}
For any $x \in R^m$ and $0<p \leq 1$, then we have that $\|x\|_p \leq\|x\|_0^{\frac{1}{p}-\frac{1}{2}}\|x\|_2$
\end{lemma}
\begin{proof}
For any $x \in R^m$, without loss of generality, we can rearrange the elements of $x$ such that $x_i=0$ $(i \in \{\|x\|_0+1,\ldots,m\})$. According to H\"older inequality, we can show that
$$\|x\|_p^p=\sum_{i=1}^{\|x\|_0} |x_i|^p \leq (\sum_{i=1}^{\|x\|_0} (|x_i|^p)^{\frac{2}{p}})^{\frac{p}{2}}(\sum_{i=1}^{\|x\|_0} 1)^{1-\frac{p}{2}}=\|x\|_0^{1-\frac{p}{2}}(\|x\|_2^2)^{\frac{p}{2}}=\|x\|_0^{1-\frac{p}{2}} \|x\|_2^p$$
that is $\|x\|_p \leq\|x\|_0^{\frac{1}{p}-\frac{1}{2}}\|x\|_2$.
\end{proof}
\begin{lemma}
Given a matrix $A \in R^{n \times m}$. If $u\|x\|_2 \leq \|Ax\|_2 \leq w\|x\|_2$ holds for any $\|x\|_0 \leq 2t$.
 then we have that $$\displaystyle |\langle Az_1,Az_2\rangle| \leq \frac{w^2-u^2}{2}\|z_1\|_2\|z_2\|_2$$ for any $z_1$ and $z_2$,with $\|z_i\|_0 \leq t\ (i=1,2)$ and $support(z_1) \cap support(z_2)= \varnothing$
\end{lemma}
\begin{proof}
According to the assumption on matrix $A$, it is easy to get that
$$|\langle Az_1,Az_2 \rangle|=\frac{1}{4}\left|\|Az_1+Az_2\|_2^2-\|Az_1-Az_2\|_2^2\right| \leq \frac{1}{4}(w^2\|z_1+z_2\|_2^2-u^2\|z_1-z_2\|_2^2) )$$
Since $support(z_1) \cap support(z_2)= \varnothing$, we have that $$\|z_1+z_2\|_2^2=\|z_1-z_2\|_2^2=\|z_1\|_2^2+\|z_2\|_2^2$$
From which we get that $$\displaystyle|\langle Az_1,Az_2 \rangle|\leq \frac{1}{4}(w^2-u^2)(\|z_1\|_2^2+\|z_2\|_2^2) \leq \frac{w^2-u^2}{2}\|z_1\|_2\|z_2\|_2$$
\end{proof}
With the above lemmas in hand, we now can prove our main theorems
\begin{theorem}
Given a matrix $A \in R^{n \times m}$ with $n \leq m$. If $x^{\ast}$ is the unique solution to $l_0$-minimization (1), for an arbitrary index set $\tau_0 \subset \{1,2,\ldots,m\}$ with $|\tau_0|=\|x^{\ast}\|_0=t$ and for any $x \in Ker(A)$, we have the following inequality:
\begin{eqnarray}
\|x_{\tau _0}\|_p \leq h^{*}(p,A,t)\|x_{\tau _0^C}\|_p
\end{eqnarray}
where $h^{*}(p,A,t)=\displaystyle\frac{\sqrt2 +1}{2}\left(\frac{t}{t+1}\right)^{\frac{1}{p}}\left[\frac{(\lambda -1)(m-2-t)}{2t}+\left(\lambda +\sqrt{\frac{1}{t+1}}\right)t^{- \frac{1}{2}}\right]$\\ and $\lambda=\displaystyle\frac{\lambda_{max}(A^TA)}{\lambda_{min^{+}}(A^TA)}$
\end{theorem}
\begin{proof}
According to Theorem 1 and Corollary 3, it is easy to get that $\|x\|_0 \geq 2t+1$ for any $x \in Ker(A)$.

According to Lemma 1 and Corollary 4, we can find constants $u$ and $w$ such that $$u\|z\|_2 \leq \|Az\|_2 \leq w\|z\|_2$$ for any $\|z\|_0 \leq 2t$. Now we consider a vector $x \in Ker(A)$, and consider an arbitrary index set $\tau_0 \subset \{1,2,\ldots,m\}\ \ with\ |\tau_0|=t$. We partition the complement of $\tau_0$

$\tau_1$=\{indices of the largest $t+1$ absolute-values component of $x$ except $\tau_0$ \}

$\tau_2$=\{indices of the largest t absolute-values component of $x$ except $\tau_0$ and $\tau_1$\}

$\tau_3$=\{indices of the largest t absolute-values component of $x$ except $\tau_0$ , $\tau_1$\ and  $\tau_2$\}

\ldots

$\tau_k$=\{indices of the rest component of $x$\}

Obviously, the set $\tau_i\ (i=2\ldots k-1)$ has $t$ elements except $\tau_k$ possibly. It is obvious that every element in the set $\tau_1$ is nonzero, so we consider that

$\tau_1^{(1)}$=\{indices of the $t$ largest absolute-values components of $\tau_1$ \}

$\tau_1^{(2)}$=\{indices of the rest components of $\tau_1$ \}. Such that $\tau_1=\tau_1^{(1)} \cup \tau_1^{(2)}$

Since $u\|z\|_2 \leq \|Az\|_2 \leq w\|z\|_2$ for any $\|z\|_0 \leq 2t$. we have that
\begin{eqnarray}
\displaystyle\notag\|x_{\tau_0}\|_2^2+\|x_{\tau_1}\|_2^2&=&\|x_{\tau_0}\|_2^2+\|x_{\tau_1^{(1)}}\|_2^2+\|x_{\tau_1^{(2)}}\|_2^2 \\
\notag&=&\|x_{\tau_0}+x_{\tau_1^{(1)}}\|_2^2+\|x_{\tau_1^{(2)}}\|_2^2\\
&\leq& \displaystyle\frac{1}{u^2}\|A(x_{\tau_0}+x_{\tau_1^{(1)}})\|_2^2+\|x_{\tau_1^{(2)}}\|_2^2
\end{eqnarray}
Since $x=x_{\tau_0}+x_{\tau_1^{(1)}}+x_{\tau_1^{(2)}}+x_{\tau_2}+\ldots+x_{\tau_k} \in Ker(A)$, we have that
\begin{eqnarray}
\notag\|A(x_{\tau_0}+x_{\tau_1^{(1)}})\|_2^2&=&\langle A(-x_{\tau_0}-x_{\tau_1^{(1)}},A(x_{\tau_1^{(2)}}+x_{\tau_2}+\ldots+x_{\tau_k})\rangle\\
 \notag&=&\langle A(-x_{\tau_0}-x_{\tau_1^{(1)}}),Ax_{\tau_1^{(2)}}\rangle +\\
 & &\displaystyle\sum_{i=2}^k(\langle A(-x_{\tau_0}),Ax_{\tau_1}\rangle+\langle A(-x_{\tau_1}^{(1)}),Ax_{\tau_1}\rangle)
\end{eqnarray}
According to Lemma 3, we get that
\begin{eqnarray}
\begin{split}
&\langle A(-x_{\tau_0}),Ax_{\tau_i}\rangle \leq \frac{w^2-u^2}{2}\|x_{\tau_0}\|_2\|x_{\tau_i}\|_2\\
&\langle A(-x_{\tau_1}^{(1)}),Ax_{\tau_i}\rangle \leq \frac{w^2-u^2}{2}\|x_{\tau_1}^{(1)}\|_2\|x_{\tau_i}\|_2
\end{split}
\end{eqnarray}
substituting the inequalities (9) into (8), we have
\begin{eqnarray}
\notag \|A(x_{\tau_0}+x_{\tau_1^{(1)}})\|_2^2 & \leq &\|A(-x_{\tau_0}-x_{\tau_1^{(1)}})\|_2\|Ax_{\tau_1}^{(2)}\|+\\
\notag & & \frac{w^2-u^2}{2}(\sum_{i=2}^k\|x_{\tau_i}\|_2)(\|x_{\tau_0}\|_2+\|x_{\tau_1^{(1)}}\|_2)\\
\notag & \leq & w^2(\|x_{\tau_0}\|_2+\|x_{\tau_1^{(1)}}\|_2)\|x_{\tau_1^{(2)}}\|_2+\\
& &\frac{w^2-u^2}{2}(\sum_{i=2}^k\|x_{\tau_i}\|_2)(\|x_{\tau_0}\|_2+\|x_{\tau_1^{(1)}}\|_2)
\end{eqnarray}
By the definition of $x_{\tau_1^{(1)}}$ and $x_{\tau_1}$, it is easy to get that $\|x_{\tau_1^{(1)}}\|_2 \leq \|x_{\tau_1}\|_2$ and $\displaystyle \|x_{\tau_1^{(2)}}\|_2 \leq \sqrt{\frac{1}{t+1}}\|x_{\tau_1}\|_2$, and hence,
\begin{eqnarray}
\displaystyle\|x_{\tau_1^{(2)}}\|_2^2 \leq (\|x_{\tau_0}\|_2+\|x_{\tau_1}\|_2)\sqrt{\frac{1}{t+1}}\|x_{\tau_1^{(2)}}\|
\end{eqnarray}
Substituting the inequalities (11) and (10) into (7)
\begin{eqnarray}
\displaystyle\notag\|x_{\tau_0}\|_2^2+\|x_{\tau_1}\|_2^2 & \leq & \displaystyle\frac{1}{u^2}\|A(x_{\tau_0}+x_{\tau_1^{(1)}})\|_2^2+\|x_{\tau_1^{(2)}}\|_2^2\\
\notag& \leq &  (\|x_{\tau_0}\|_2+\|x_{\tau_1}\|_2)(\displaystyle\frac{w^2-u^2}{2u^2} \displaystyle\sum_{i=2}^k \|x_{\tau_i}\|_2+\\
&  &\left(\frac{w^2}{u^2}+\sqrt{\frac{1}{t+1}}\right)\|x_{\tau_1^{(2)}}\|_2)
\end{eqnarray}

For any $i \geq 2$ and any element $a$ of $x_{\tau_i}$, it is easy to get that $|a|^p \leq \displaystyle\frac{1}{t+1}\|x_{\tau_1}\|_p^p$, so that $\|x_{\tau_i}\|_2^2 \leq t(t+1)^{-\frac{2}{p}}\|x_{\tau_1}\|_p^2$ and $\|x_{\tau_i}^{(2)}\|_2^2 \leq (t+1)^{-\frac{2}{p}}\|x_{\tau_1}\|_p^2$

Substituting the inequalities into (12), we can derive that
\begin{eqnarray}
\displaystyle\notag
\|x_{\tau_0}\|_2^2+\|x_{\tau_1}\|_2^2 & \leq & (\|x_{\tau_0}\|_2+\|x_{\tau_1}\|_2)(\frac{w^2-u^2}{2u^2}t^{\frac{1}{2}}(t+1)^{- \frac{1}{p}}(k-1)+\\
&  &(\frac{w^2}{u^2}+\sqrt{\frac{1}{t+1}})(t+1)^{-\frac{1}{p}})\|x_{\tau_1}\|_p
\end{eqnarray}

We denote $\displaystyle r=\frac{w^2}{u^2}$, and $\displaystyle B=\left(\frac{w^2-u^2}{2u^2}t^{\frac{1}{2}}(t+1)^{ -\frac{1}{p}}(k-1)+(\frac{w^2}{u^2}+\sqrt{\frac{1}{t+1}})(t+1)^{-\frac{1}{p}}\right)\|x_{\tau_1}\|_p$ \\then we can get such an inequality, $$\|x_{\tau_0}\|_2^2+\|x_{\tau_1}\|_2^2 \leq B(\|x_{\tau_0}\|_2+\|x_{\tau_1}\|_2)$$ and $$\left(\|x_{\tau_0}\|_2-\displaystyle\frac{B}{2}\right)^2-\left(\|x_{\tau_1}\|_2-\displaystyle\frac{B}{2}\right)^2\leq \displaystyle\frac{B^2}{2}$$ therefore, $\|x_{\tau_0}\|_2 \leq \displaystyle\frac{\sqrt2+1}{2}B$

According to Lemma 2, we have that: $$\displaystyle \|x_{\tau_0}\|_p \leq t^{\frac{1}{p}-\frac{1}{2}}\|x_{\tau_0}\|_2 \leq t^{\frac{1}{p}-\frac{1}{2}}(\sqrt2+1)\displaystyle\frac{B}{2}$$

Substituting B into this inequality, we can obtain that
\begin{eqnarray}
\displaystyle\notag\|x_{\tau_0}\|_p & \leq & t^{\frac{1}{p}-\frac{1}{2}}\left(\frac{\sqrt2+1}{2}\right)\left(\frac{r-1}{2}t^{\frac{1}{2}}(t+1)^{-\frac{1}{p}}(k-1)+
(r+\sqrt{\frac{1}{t+1}})(t+1)^{-\frac{1}{p}}\right)\|x_{\tau_1}\|_p\\
& \leq & \left({\frac{t}{t+1}}\right)^{\frac{1}{p}}\left(\frac{\sqrt2+1}{2}\right)\left(\frac{r-1}{2}(k-1)+
(r+\sqrt{\frac{1}{t+1}})t^{-\frac{1}{2}}\right)\|x_{\tau_1}\|_p
\end{eqnarray}
We notice that the sets $\tau_0$ and $\tau_i$($i=2\ldots k-1$) all have $t$ elements and the set $\tau_1$ has $t+1$ elements, such that $kt+2 \leq m \leq (k+1)t+1$, we can get that $k \leq \displaystyle\frac{m-2}{t}$

According to Lemma 1, we have that $r=\displaystyle\frac{w^2}{u^2} \leq \lambda =\displaystyle\frac{\lambda_{max}(A^TA)}{\lambda_{min}(A^TA)}$. Substituting the inequalities into (14), we can obtain
\begin{eqnarray}
\notag\|x_{\tau_0}\|_p &\leq & \displaystyle\frac{\sqrt2 +1}{2}\left(\frac{t}{t+1}\right)^{\frac{1}{p}}\left(\frac{(\lambda -1)(m-2-t)}{2t}+\left(\lambda +\sqrt{\frac{1}{t+1}}\right)t^{- \frac{1}{2}}\right)\|x_{\tau_1}\|_p\\
& \leq &h^{\ast}(p,A,t)\|x_{\tau_0^C}\|_p
\end{eqnarray}

Therefore, the proof is completed.
\end{proof}
Theorem 3 presents a result which is very similar to the result in Theorem 1. However, it is worth to being pointed out that the constant $h^{\ast}(p,A,t)$ plays a central role in theorem 3. In fact, we can treat $h^{\ast}(p.A,t)$ as an estimation to $h(p,A,t)$ in Definition 1, where the former is calculateble and while the latter is NP-hard, so $h^{\ast}(p,A,t)$ may be considered as an improvement of $h(p,A,t)$. According to Theorem 1, if we take $t$ as the $l_0$-norm of the unique solution to $l_0$-minimization, then we can get the main contribution as soon as the inequality $h^{\ast}(p,A,t)<1$ is satisfied.
\begin{theorem}
Assume $A \in R^{n \times m}$ is an underdetermined matrix of full rank, and denote $S^*=|support(A^T(AA^T)^{-1}b)|$. If $l_0$-minimization (1) has an unique solution, then $l_p$-minimization\ (4) has the same unique solution as $l_0$-minimization (1) for any $0<p<p^*(A,b)$, where $\displaystyle p^*(A,b)=max\{h(S^*),h(\lceil \frac{m-2.5}{2} \rceil+1)\}$
$$\displaystyle h(x)=\frac{ln (x+1)-ln x}{ln \left[ \left(\displaystyle\frac{\sqrt2+1}{2}\right) \left[\displaystyle\frac{(\lambda-1)(m-3)}{2}+\lambda+\displaystyle\sqrt\frac{1}{2}\right] \right]}\
and\ \lambda=\displaystyle\frac{\lambda_{max}(A^TA)}{\lambda_{min^{+}}(A^TA)}$$
\end{theorem}

\begin{proof}
 According to Theorem 3 and Theorem 1, we can get the equivalence between $l_0$-minimization (1) and $l_p$-minimization (4) as soon as $h^{\ast}(p,A,t)<1$ is satisfied. However $t$ can't be calculated directly. We need to estimate $t$ and change the inequality $h^{\ast}(p,A,t)<1$ into a computable one through inequality technique.

 Due to the integer-value virtue of $\|x\|_0$, we can have that$$(\lambda+\sqrt{\displaystyle\frac{1}{t+1}})t^{-\frac{1}{2}} \leq \lambda +\sqrt{\displaystyle\frac{1}{2}}$$ and $$\displaystyle\frac{(\lambda-1)(m-2-t)}{2t}\leq \displaystyle\frac{(\lambda-1)(m-3)}{2}$$
Therefore,$$\displaystyle h^{\ast}(p,A,t) \leq \frac{\sqrt2+1}{2}(\frac{t}{t+1})^{\frac{1}{p}}(\frac{(\lambda-1)(m-3)}{2}+\lambda+\sqrt{\frac{1}{2}})$$.

\begin{figure}[h]
\centering
\includegraphics[width=5.00in,height=3.00in]{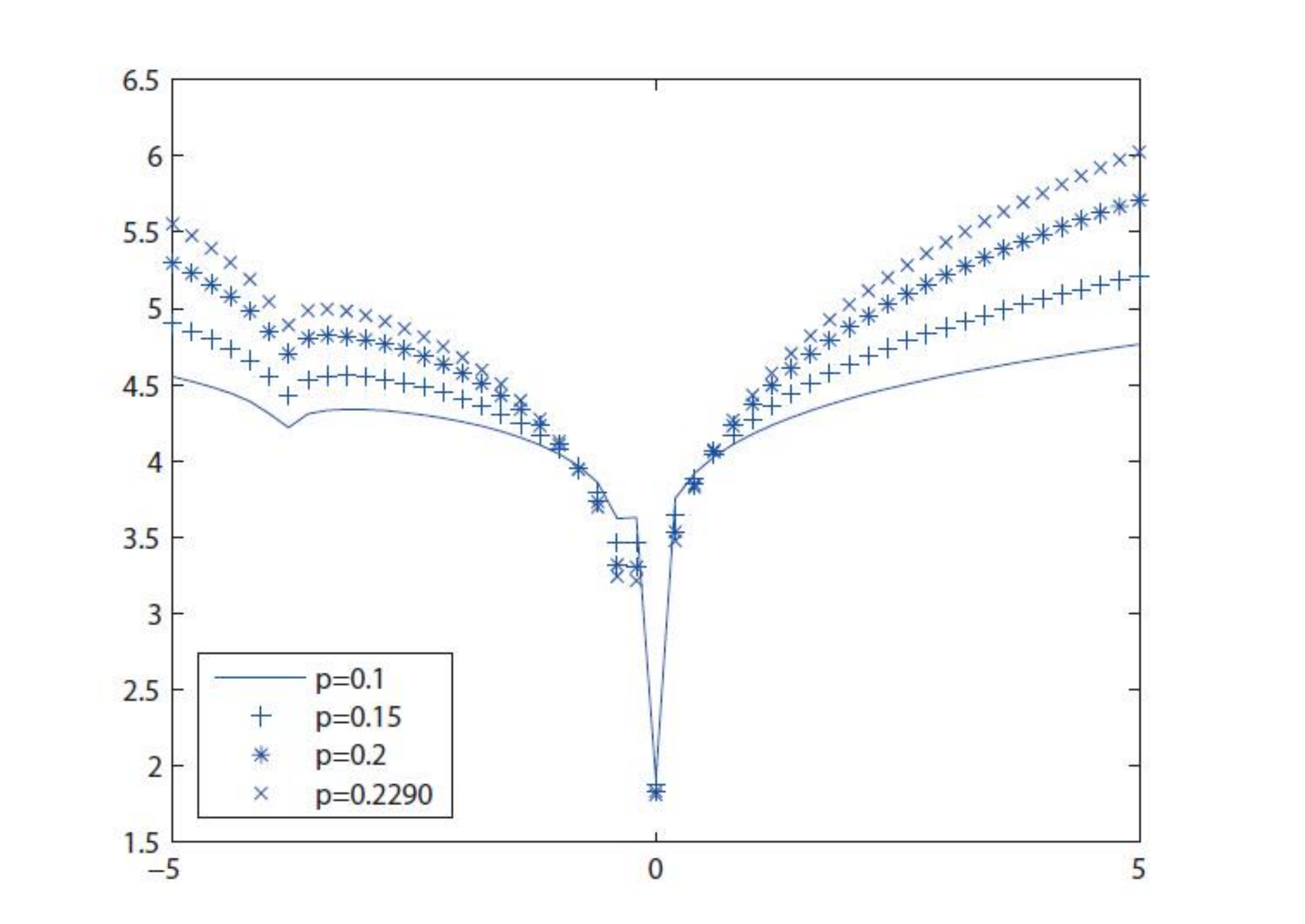}
\caption{the p-norm of the solutions to the different cases in Example 1}
\end{figure}
 According to Corollary 3, we have that $t \leq \lceil \frac{m-2.5}{2} \rceil+1$. It is obvious that $x=A^T(AA^T)^{-1}b$ is a solution to the underdetermined system $Ax=b$ so that $t \leq |S^{\ast}|$ where $S^{\ast}=|support(A^T(AA^T)^{-1}b)|$

 Furthermore, it is easy to prove that the function $$\displaystyle f(t,p)=\frac{\sqrt2+1}{2}(\frac{t}{t+1})^{\frac{1}{p}}(\frac{(\lambda-1)(m-3)}{2}+\lambda+\sqrt{\frac{1}{2}})$$ is an increasing function in $t$ when $p$ is fixed. Meanwhile $f(t,p)$ is a decreasing function in $p$ when $t$ is fixed

Since $f(t,p) \leq min\{ f(\lceil \frac{m-2.5}{2} \rceil+1 ,p), f(S^{\ast} ,p)\}$, we can get $f(t,p)<1$ if one of these two inequalities $f(\lceil \frac{m-2.5}{2} \rceil+1 ,p)<1$ or $f(S^{\ast} ,p)\}<1$ is satisfied.

Furthermore, the inequality $f(x,p)<1$ when $x$ fixed is very easy to solve, the range of such $p$ is
$$p<h(x)=\displaystyle\frac{ln (x+1)-ln x}{ln \left[(\displaystyle\frac{\sqrt2+1}{2})[\displaystyle\frac{(\lambda-1)(m-3)}{2}+\lambda+\displaystyle\sqrt\frac{1}{2}]\right]}$$

Hence, for any $0<p<p^{\ast}=max\{h(S^{\ast}),h(\lceil \frac{m-2.5}{2} \rceil+1)\}$, we have that $h^{\ast}(p,A,t) \leq f(t,p)<1$. Therefore, by Theorem 1, we know both $l_0$-minimization (1) and $l_p$-minimization (4) have the same unique solution.
\end{proof}
Combining Theorems 3 and 4, we have reached the major goals of this paper. The most important result in these two theorems is the analysis expression of $p^{\ast}(A,b)$, with which, the specific range of $p$ can be calculated easily.
\begin{figure}[h]
\centering
\includegraphics[width=5.00in,height=3.00in]{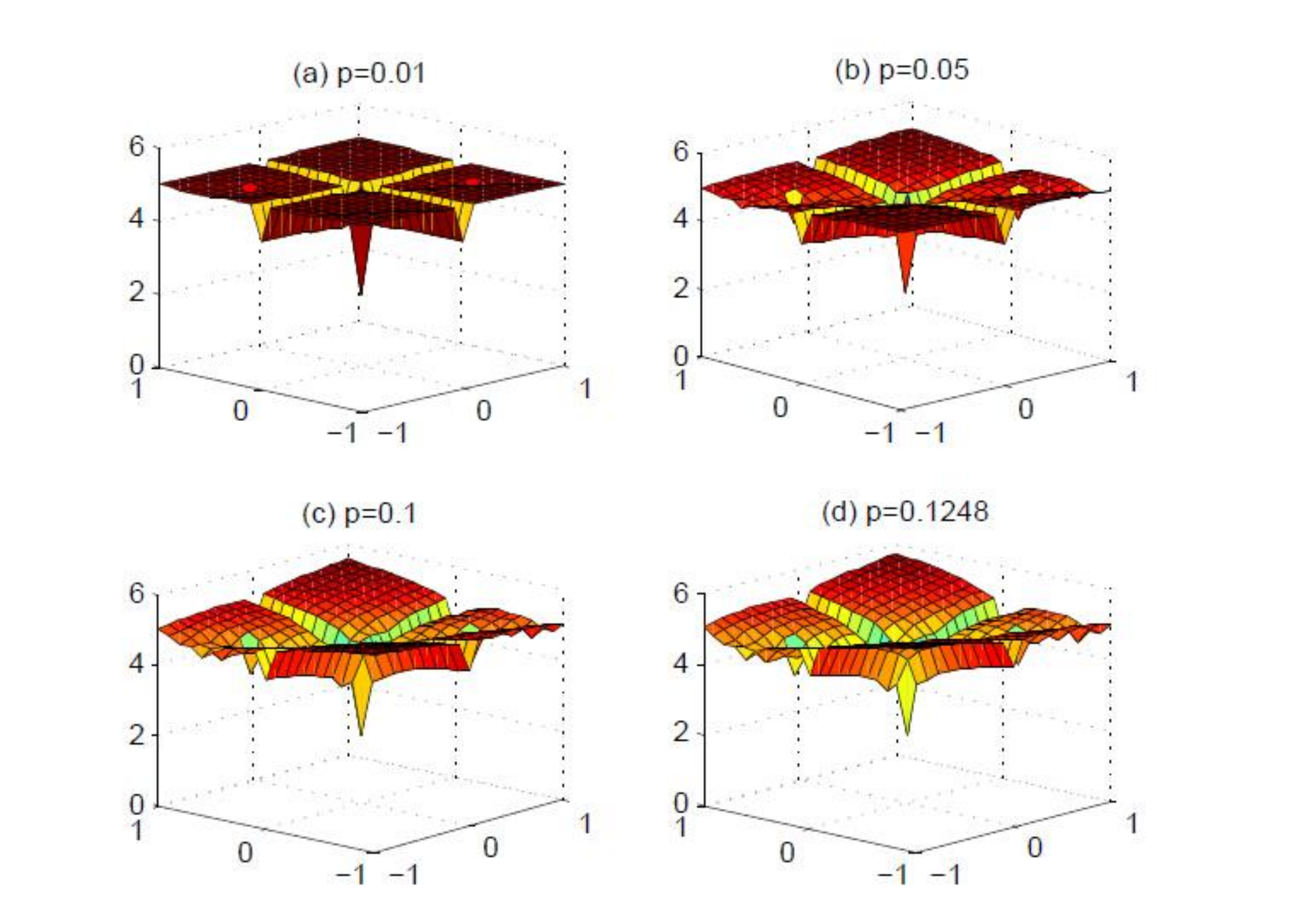}
\caption{the p-norm of the solutions to the different cases in Example 2}
\end{figure}
We present two examples to demonstrate the validation of Theorem 4, 
\begin{example}
We consider an underdetermined system $Ax=b$, where
\begin{center}
$A=\left(
\begin{array}{cccc}
1 & 1.5 & 0.7 & 0 \\
0 & 2 & 0.5 & 1 \\
1 & 0.5 & 1 & 1 \\
\end{array}
\right)$ and $b=\left(
\begin{array}{ccc}
1.65\\
1.4\\
0.95
\end{array}
\right)$
\end{center}
\end{example}

 It is obvious that the sparest solution is $x^{\ast}=[0.6\ 0.7\ 0\ 0]^T$, and the solution to the equation can be expressed as the following form:
 $$x=[0.6\ 0.7\ 0\ 0\ ]^T+t[\frac{18}{11}\ \frac{2}{11}\ -\frac{30}{11}\ 1]^T$$
 Therefore, the $l_p$-norm of $x$ can be expressed $$\|x\|_p^p=|0.6+\frac{18}{11}t|^p+|0.7+\frac{2}{11}t|^p+|\frac{30}{11}t|^p+|t|^p$$

Furthermore, $\lambda_{max}(A^TA)=3.1286$,$\lambda_{min}(A^TA)=0.9342$ and $\lambda=\frac{\lambda_{max}(A^TA)}{\lambda_{min}(A^TA}=11.2155$. It is easy to get that $A^T(AA^T)^{-1}b=[0.4372\
0.6819\ 0.2773\ -0.0995]^T$, hence $h(S^{\ast})=0.0738$ and $h(\lceil \frac{m-2.5}{2} \rceil+1)=0.2293$, so $p^{\ast}=0.2293$.

As shown in the Fig1, we can get the solution to $l_p$-minimization in different cases when $p=0.1,0.15,0.2,\ and\ 0.2290$ for $t=0$, which are just the sparest solution $x^{\ast}=[0.6\ 0.7\ 0\ 0]^T$

\begin{figure}[h]
\centering
\includegraphics[width=5.00in,height=3.00in]{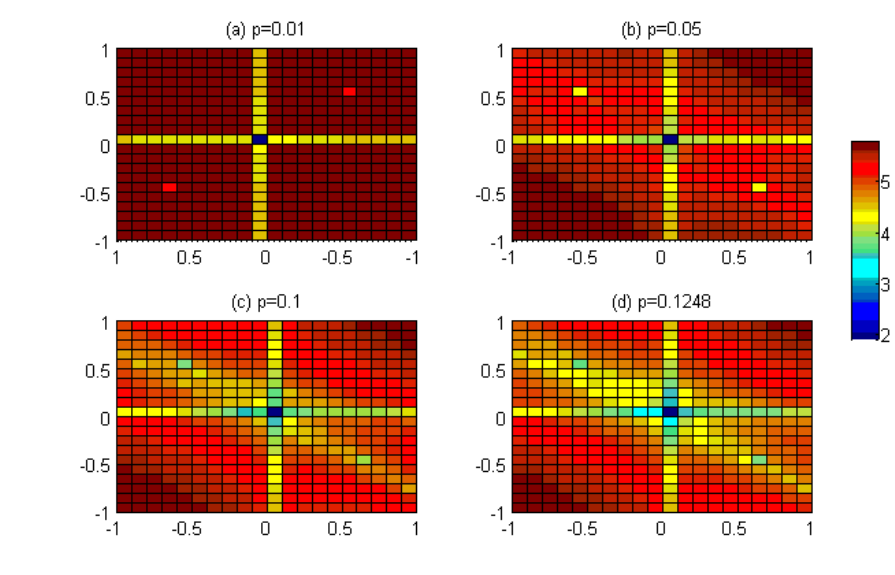}
\caption{coordinates on s-t plane for the p-norm to the solution to the different cases in Example 2}
\end{figure}
\begin{example}
We consider an underdetermined system $Ax=b$, where
\begin{center}
$A=\left(
\begin{array}{ccccc}
1 & 0 & 3.5 & 3 & 2.7 \\
0 & 2 & 0 & 1.5 & 4.5 \\
2 & 2 & 4 & 0.5 & 1.5
\end{array}
\right)$ and $b=\left(
\begin{array}{ccc}
1\\
1\\
3
\end{array}
\right)$
\end{center}
\end{example}

It is easy to get the sparest solution $x^{\ast}=[1\ 0.5\ 0\ 0\ 0]^T$, and the solutions of the underdetermined system $Ax=b$ can be expressed as the following parameterized form
$$x=[1\ 0.5\ 0\ 0\ 0]^T+s[\frac{31}{6}\ -\frac{3}{4}\ -\frac{7}{3}\ 1\ 0]^T+t[7.1\ -2.25\ -2.8\ 0\ 1]^T \ where \ s,t\in R$$
Therefore,
$$\|x\|_p^p=|1+\frac{31}{6}s+7.1t|^p+|0.5-\frac{3}{4}s-2.25t|+|-\frac{7}{3}s-2.8t|^p+|s|^p+|t|^p$$
Furthermore, $\lambda_{max}(A^TA)=7.8244$, $\lambda_{min}(A^TA)=2.3929$ and $\lambda=\frac{\lambda_{max}(A^TA)}{\lambda_{min}(A^TA}=10.7818$. It is easy to get that $A^T(AA^T)^{-1}b=[0.2851\
0.4924\ 0.3600\ -0.2639\ 0.0941]^T$, hence $h(S^{\ast})=0.0562$ and $h(\lceil \frac{m-2.5}{2} \rceil+1)=0.1249$, so $p^{\ast}=0.1249$.

From Fig2, we can also find that the solution in different case when $p=0.01,0.05,0.1,\ and\ 0.1248$ for $s=t=0$. The result can be seen more clearly in Fig3.

\section{CONCLUSION}
 In this paper we have studied the equivalence between $l_0$-minimization (1) and $l_p$-minimization (4). By using the null space property, a sufficient and necessary condition to recover a sparse vector via these two models, we present an analysis expression of $p^{\ast}(A,b)$ in Theorem 4 to make sure both $l_0$-minimization (1) and $l_p$-minimization (4) for $0<p\leq p^{\ast}(A,b)$ have the same unique solution.

However, in this paper, we only consider the situation when $l_0$-minimization (1) has the unique solution. and it should be pointed out that the condition of uniqueness is a very important one for our main contribution, especially for Lemma 1. Under the condition of uniqueness, in Lemma 1, we can get a result which is similar to RIP. Unlike the nonhomogeneity of RIP, the result is not conflict with the equivalence of all the linear systems $aAx=ax,\ a\in R$. Therefore, the authors guess that we can replace RIP with the condition of uniqueness in more application.

Meanwhile, it also need to be pointed out that the main result in Theorem 4, the analysis expression of $p^{\ast}(A,b)$ have a closer relationship with the matrix $A$ than the vector $b$. Let $p_{new}(A)=h(\lceil \frac{m-2.5}{2} \rceil)$, we can find that, although $p_{new}(A)$ is smaller than $p^{\ast}(A,b)$, $p_{new}(A)$ have the same effect as $p^{\ast}(A,b)$ does. The phenomenon may also be explained by the condition of uniqueness and the null space property, a property only depend on the structure of the matrix $A$ itself. Therefore, what is such $p^{\ast}(A,b)$ without the condition of uniqueness? The authors think the answer to this problem will be an important improvement for the application of $l_p$-minimization. In conclusion, the authors hope that in publishing this paper, a brick will be thrown out and be replaced with a gem.

 \section{Acknowledgements}
 This work was supported by the NSFC under contact no.11131006 and by the EU FP7 Project EYE2E (269118) and LIVCODE (295151), and in part by the National Basic Research Program of China under contact no. 2013CB329404

\end{document}